\documentclass{lmcs}

\usepackage{amsmath,amssymb,mathtools}
\usepackage{booktabs}
\usepackage{hyperref}

\newcommand{\LAc}{\mathrm{LA}_{\!c}}
\newcommand{\GLUE}{\mathsf{GLUE}}
\newcommand{\SELECT}{\mathsf{SELECT}}
\newcommand{\SEARCHSEL}{\mathsf{SEARCH}_{\mathsf{SEL}}}
\newcommand{\FPS}{\mathrm{FP}_{\mathrm{search}}}
\newcommand{\FNPS}{\mathrm{FNP}_{\mathrm{search}}}
\newcommand{\RGLUE}{R_{\mathsf{GLUE}}}
\newcommand{\RSEL}{R_{\mathsf{SEL}}}

\newcommand{\eff}{\mathrm{eff}}
\newcommand{\Res}{\mathrm{Res}}

\title[Computational Free Will and P versus NP]{Computational Free Will as Global Selection:\
From Sheaf-Theoretic Gluing to a Conditional Separation of P and NP}

\author[J. Clech]{J\'er\^ome Clech}[a,b]
\address{Colonel, PhD, Habilitation to Supervise Research (HDR); Chairholder of the Chair of Applied Air and Space Strategies, Centre for Aerospace Strategic Studies (CESA), French Air and Space Force, Paris, France}
\address{Associate Researcher, Technology and Global Affairs Innovation Hub, Paris School of International Affairs (PSIA), Sciences Po, Paris, France}
\email{jerome.clech@sciencespo.fr}

\begin{document}

\begin{abstract}
We formalise computational free will by separating locally constrained admissibility from the selection of one global continuation.  Global sections of a finite choice presheaf form an admissible set, GLUE; SELECT singles out the continuation realised at a pre-identified occurrence.  A uniform trace relation certifies that continuation efficiently after the act, although it is assumed not to be uniformly anticipable in polynomial time from the prior occurrence input.  Under explicit uniformity, balance, historical-completeness, and unique-projection assumptions, this trace defines a total FNP search relation with no deterministic polynomial-time selector.  Thus existence of computational free will in the stated sense implies a separation between polynomially verifiable and polynomially solvable search, and hence that P differs from NP.  The result is conditional and gives no unconditional class separation.
\end{abstract}

\maketitle


\section{Introduction}\label{sec:introduction}

Strategic foresight does more than extrapolate trends or arrange possible futures.  It first determines how a situation is made intelligible: it chooses categories, delimits a problem, and changes the representational system in which solutions can be conceived.  A change in representation therefore changes the space of admissible continuations.  This observation motivates the present mathematical question.  When a prior informational state constrains several coherent global continuations, what distinguishes the existence of those continuations from the act that makes one of them actual?

The proposed answer separates two operations.  The first, \(\GLUE\), constructs or characterises the global continuations compatible with all local constraints.  The second, \(\SELECT\), singles out the continuation that is actually realised.  This distinction is elementary but consequential: a full description of the admissible solution set does not, by itself, determine which member will become the historical outcome.  We use sheaf-theoretic language to express the local-to-global structure, a finite constraint-satisfaction model to make it effective, and a trace relation to formalise the asymmetry between anticipation before the act and verification after it.

The resulting notion of \emph{computational free will}, denoted by \(\LAc\), is not intended as a definition of metaphysical free will.  It is a property of an encoded family of choice processes.  The prior state determines a structured set of at least two admissible global continuations; an occurrence identifier is fixed without encoding its outcome; one continuation is realised; the realised value is certifiable a posteriori by a uniform polynomial-time relation; and no uniform polynomial-time procedure predicts it from the prior occurrence input, even after admissible polynomial-time changes of representation.

The main result is deliberately conditional.  Under the formal assumptions made precise below, a family satisfying \(\LAc\) induces a polynomially balanced and polynomial-time decidable search relation.  Unique projection of accepted certificates onto the realised continuation ensures that any deterministic polynomial-time selector for that relation computes \(\SELECT\).  The non-anticipability condition therefore excludes such a selector.  With the search-class notation defined in \autoref{sec:main-result}, we obtain
\[
  \LAc \quad\Longrightarrow\quad \FPS\neq\FNPS
  \quad\Longrightarrow\quad \mathrm{P}\neq\mathrm{NP}.
\]
The first implication is the contribution formalised here; the second is the standard search-to-decision consequence under the usual conventions for polynomially balanced NP search relations \cite{Papadimitriou1994,AroraBarak2009}.

Two safeguards are central.  First, the sheaf-theoretic motivation is not presented as a direct proof that \(\mathrm{P}\neq\mathrm{NP}\).  Equality of the decision classes would not automatically induce a Morita equivalence between arbitrary classifying topoi, nor is a gluing obstruction identical to a polynomial-time lower bound.  Second, the Tseitin construction is used only as a restricted witness showing that local-to-global obstruction can yield exponential resolution size on a suitable bounded-degree expander family \cite{Tseitin1968,Urquhart1987,BenSassonWigderson2001}.  The general conditional theorem does not infer universal non-anticipability from that restricted lower bound.

The paper is organised as follows.  \autoref{sec:sheaf-chapeau} gives the sheaf-theoretic conceptual framework.  \autoref{sec:lac} defines computational free will independently of the complexity-class conclusion.  \autoref{sec:finite-model} identifies global sections with solutions of a finite CSP.  \autoref{sec:glue-select} separates gluing from selection, and \autoref{sec:trace} introduces retrospective certification.  \autoref{sec:tseitin} gives the restricted Tseitin witness.  \autoref{sec:representation} treats effective changes of representation.  \autoref{sec:main-result} proves the conditional separation between polynomially verifiable and polynomially solvable search and records exactly what is, and is not, established.

\section{Sheaf-theoretic conceptual framework}\label{sec:sheaf-chapeau}

This section records the conceptual argument that preceded the finite CSP and complexity-theoretic formulation.  Its purpose is to identify the mathematical structure of the problem, not to supply the final complexity proof.

\subsection{From local determinations to a global decision}

A decision is taken within prior determinations: the state and history of an agent, available reasons and information, environmental constraints, and interactions.  Each constrains the possible continuations.  The question is not whether the choice has causes, but whether the totality of prior determinations already singles out one global continuation before the act.

Let \(\mathcal U=\{U_i\}_{i\in I}\) be a cover of contexts and let \(F(U_i)\) be a set of admissible local sections.  A family \((s_i)_{i\in I}\), with \(s_i\in F(U_i)\), is compatible when
\begin{equation}\label{eq:compatibility}
  s_i\vert_{U_i\cap U_j}=s_j\vert_{U_i\cap U_j}
  \qquad\text{for all }i,j\in I.
\end{equation}
The choice problem becomes a gluing problem: do compatible local determinations admit a global section, and, if so, is that section unique?

Three regimes must be distinguished.
\begin{enumerate}
  \item \emph{Unique globalisation.}  The local data impose one global section; the continuation is already singled out by the prior constraints.
  \item \emph{Obstruction to gluing.}  No global section exists.  When an appropriate cohomology theory is available, a non-zero obstruction class may witness this failure.
  \item \emph{Plurality of globalisations.}  At least two global sections remain compatible with the same prior determinations.
\end{enumerate}
The third regime is the relevant one for computational choice.  It differs from both causal absence and inconsistency: the prior state determines a structured possibility space without yet determining its historical actualisation.  Sheaf-theoretic analyses of contextuality provide a related use of global-section obstructions, although the interpretation and target here are different \cite{AbramskyBrandenburger2011,AbramskyMansfieldBarbosa2012}.

For a prior state \(x\), let \(F_x\) encode the corresponding determinations and compatibilities.  The admissible global continuations are
\begin{equation}\label{eq:multiple-global-sections}
  \Gamma(F_x)=\{s^{(1)},s^{(2)},\ldots,s^{(k)}\},\qquad k\geq 2.
\end{equation}
The act introduces a further operation
\begin{equation}\label{eq:conceptual-selection}
  \Gamma(F_x)\longrightarrow s^*\in\Gamma(F_x).
\end{equation}
The map in \eqref{eq:conceptual-selection} is not the gluing operation itself.  The sheaf organises admissibility; the act singles out one admissible global continuation.

\subsection{What the conceptual argument does and does not prove}

Topos theory is useful here because it treats local objects, restriction, compatibility, gluing, and changes of presentation in one framework \cite{MacLaneMoerdijk1992}.  It also forces a distinction between an intrinsic obstruction and an artefact of representation.  However, the conceptual argument stops at the following statement.

\begin{prop}[Conceptual local-to-global structure]\label{prop:conceptual}
If prior determinations are represented by local sections whose compatible globalisations form a set of cardinality at least two, then the prior sheaf structure does not by itself single out the realised continuation.  An additional selection operation is required to pass from global admissibility to actualisation.
\end{prop}

The proposition says nothing yet about polynomial time.  In particular, neither a categorical equivalence nor the existence of an obstruction automatically gives a lower bound in a computational model.  The remainder of the paper therefore constructs a finite model, separates \(\GLUE\) from \(\SELECT\), specifies certification and uniformity, and only then invokes standard complexity classes.

\section{Computational free will}\label{sec:lac}

\begin{defi}[Computational free will]\label{def:lac}
An encoded family of choice processes satisfies \emph{computational free will}, written \(\LAc\), if an instance is a pair \(q=(x,e)\), where \(x\) is the prior informational state and \(e\) is an occurrence identifier fixed before the act, and the following six conditions hold.
\begin{description}
  \item[(D)] \emph{Prior determinations.} The component \(x\) encodes a non-trivial prior informational state and its constraints; \(e\) identifies the occurrence without encoding its outcome.
  \item[(P)] \emph{Admissible plurality.} The prior state determines a set \(\GLUE(x)\) of globally coherent continuations with \(\lvert\GLUE(x)\rvert\geq2\).
  \item[(G)] \emph{Global coherence.} Members of \(\GLUE(x)\) satisfy all local compatibility and global admissibility conditions; they are not arbitrary alternatives.
  \item[(S)] \emph{Singularisation.} The occurrence \(e\) realises one well-defined value \(s_e=\SELECT(q)\in\GLUE(x)\).
  \item[(T)] \emph{Traceability.} After actualisation, a finite trace \(\tau_e\) certifies, through a fixed uniform predicate, that \(s_e\) was the continuation realised at occurrence \(e\).
  \item[(A)] \emph{Non-anticipability.} No uniform polynomial-time procedure computes \(\SELECT(q)\) from the prior input \(q=(x,e)\) on the entire encoded family, including after any admissible polynomial-time change of representation.
\end{description}
\end{defi}

The identifier \(e\) is publicly registered before actualisation, has polynomially bounded length, and is generated independently of the realised continuation.  It may be a session or experiment identifier, but it is neither a hidden seed nor post-event information.  Consequently, two occurrences with the same prior state may satisfy
\[
  \SELECT(x,e_1)\neq\SELECT(x,e_2)
\]
without making \(\SELECT\) multivalued.

Thus \(\LAc=(D,P,G,S,T,A)\).  The definition excludes several weaker phenomena.  Plurality without singularisation describes an unresolved possibility space, not a choice.  Singularisation without traceability gives no mathematically usable retrospective certificate.  Randomness alone does not satisfy the definition merely by producing unpredictable outputs: it must also obey the structured admissibility and trace conditions.  Conversely, a deterministic preference function that is uniformly polynomial-time computable from \(q\) fails condition (A).

The definition also distinguishes ontological indeterminacy from computational non-anticipability.  The former concerns what is determined in reality; the latter concerns what a uniform algorithm can compute from an encoded prior state.  No inference from one notion to the other is assumed here.

Condition (A) is an operational clause in the definition of computational free will: it expresses the absence of a uniform polynomial-time anticipation procedure for the realised continuation.  It asserts neither a separation of search classes nor \(\mathrm{P}\neq\mathrm{NP}\), and it does not by itself place any search relation in \(\FNPS\).  The separation obtained below is conditional and arises only after condition (A) is combined with the independent traceability requirements---polynomial balance, polynomial-time verification, totality, and unique projection---that place the associated selection relation in \(\FNPS\).  The significance of the theorem therefore depends on whether \autoref{def:lac} captures a substantive and independently defensible property of choice processes.

\begin{prop}[Prior non-singularisation]\label{prop:prior-nonsingularisation}
For any process satisfying \(\LAc\), the prior state determines a non-singleton structured set \(\GLUE(x)\), while the realised value \(\SELECT(q)\) is not supplied by that set description or by the outcome-independent identifier \(e\).  After actualisation, a trace makes the realised value verifiable under condition (T).
\end{prop}

\section{A finite sheaf model and its associated CSP}\label{sec:finite-model}

Let \(X\) be a finite set of elementary variables.  For each \(a\in X\), let \(D_a\) be a finite domain, and let
\[
  \mathcal U=\{U_1,\ldots,U_m\},\qquad U_i\subseteq X,
  \qquad\bigcup_{i=1}^{m}U_i=X,
\]
be a finite cover.  For \(U\subseteq X\), define the assignment presheaf
\[
  E(U)=\prod_{a\in U}D_a,
\]
with restriction maps \(\rho^U_V:E(U)\to E(V)\) given by coordinate restriction whenever \(V\subseteq U\).

More formally, let \(\mathcal C_X\) be the poset category of subsets of \(X\), with an arrow \(V\to U\) when \(V\subseteq U\), equipped with the finite-cover topology in which \(\{U_j\to U\}_j\) covers \(U\) when \(\bigcup_jU_j=U\).  Coordinate restriction makes \(E:\mathcal C_X^{\mathrm{op}}\to\mathbf{Set}\) a sheaf: compatible partial assignments glue uniquely to an assignment on their union.

A finite choice model specifies, for every context \(U_i\), a constraint relation
\[
  R_i\subseteq E(U_i)=\prod_{a\in U_i}D_a.
\]
For \(W\subseteq U_i\), write
\[
  R_i|_W:=\{r|_W:r\in R_i\}\subseteq E(W).
\]
These relations induce a constrained subpresheaf \(F\subseteq E\) by
\begin{equation}\label{eq:constraint-presheaf}
  F(U):=\bigl\{t\in E(U):
  t|_{U\cap U_i}\in R_i|_{U\cap U_i}
  \text{ for every }i\bigr\}.
\end{equation}
If \(t\in F(U)\) and \(V\subseteq U\), then \(t|_V\in F(V)\), so the coordinate restrictions of \(E\) restrict to \(F\).  The constrained presheaf \(F\) need not itself satisfy the sheaf gluing axiom; when it does, it may properly be called a sheaf of choices.  None of the finite-complexity arguments below assumes that additional property.

A family \((s_i)_{i=1}^{m}\), with \(s_i\in R_i\), is compatible when
\[
  s_i\vert_{U_i\cap U_j}=s_j\vert_{U_i\cap U_j}
  \qquad\text{for all }i,j.
\]
Because \(E\) is a sheaf, every such family has a unique underlying global assignment \(s\in E(X)\).  The global constrained assignments are
\begin{equation}\label{eq:global-sections}
  \Gamma(F):=F(X)
  =\bigl\{s\in E(X):s\vert_{U_i}\in R_i
  \text{ for every }i\bigr\}.
\end{equation}

Associate with these data the constraint-satisfaction instance \(\operatorname{CSP}(F)\) having variables \(X\), domains \(D_a\), and constraint relations \(R_i\) on the contexts \(U_i\).

\begin{prop}[Global sections and CSP solutions]\label{prop:glue-csp}
There is a canonical bijection
\[
  \Gamma(F)\cong\operatorname{Sol}(\operatorname{CSP}(F)).
\]
\end{prop}

\begin{proof}
If \(s\in\Gamma(F)=F(X)\), then \(s\vert_{U_i}\in R_i\) for every \(i\), hence \(s\) satisfies every CSP constraint.  Conversely, any solution of \(\operatorname{CSP}(F)\) restricts to an element of \(R_i\) on every context and therefore belongs to \(F(X)=\Gamma(F)\).  Both descriptions preserve the underlying global assignment.
\end{proof}

\begin{defi}[Gluing set]\label{def:glue}
For a finite choice model \(F\), define
\[
  \GLUE(F):=\Gamma(F)\cong\operatorname{Sol}(\operatorname{CSP}(F)).
\]
For a prior state \(x\), we write \(\GLUE(x):=\Gamma(F_x)\).  For an occurrence input \(q=(x,e)\), set \(\GLUE(q):=\GLUE(x)\); the occurrence identifier does not alter prior admissibility.
\end{defi}

The cardinality of \(\GLUE(F)\) distinguishes obstruction, prior singularisation, and plurality:
\[
\begin{array}{rcl}
\GLUE(F)=\varnothing &:& \text{no admissible global continuation},\\
\lvert\GLUE(F)\rvert=1 &:& \text{the constraints already single out the continuation},\\
\lvert\GLUE(F)\rvert\geq2 &:& \text{several globally admissible continuations remain}.
\end{array}
\]
This correspondence makes the local-to-global problem finite and computationally explicit, but it does not yet assign any complexity bound to gluing and does not choose a member of a non-singleton solution set.  The latter operation is isolated next.

\section{From gluing to selection}\label{sec:glue-select}

For an occurrence input \(q=(x,e)\), \(\GLUE(q)=\GLUE(x)\) answers the admissibility question: which continuations satisfy all constraints?  It does not answer the historical question: which continuation is realised at \(e\)?  We therefore introduce a separate operation.

\begin{defi}[Selection operator]\label{def:select}
On a domain of occurrence inputs \(q=(x,e)\) satisfying \(\lvert\GLUE(x)\rvert\geq2\), a selection operator is a map
\begin{equation}\label{eq:select}
  \SELECT(q)=\SELECT(x,e)=s_e\in\GLUE(x).
\end{equation}
The value \(s_e\) is the continuation actually realised at occurrence \(e\).
\end{defi}

Before the act, the mathematical object made available by the constraints is the set
\[
  \GLUE(x)=\{s^{(1)},s^{(2)},\ldots,s^{(k)}\},\qquad k\geq2.
\]
The act is represented schematically by
\[
  (x,e,\GLUE(x))\xrightarrow{\;a_e\;}(x,e,s_e),
  \qquad s_e\in\GLUE(x).
\]
The arrow is not asserted to be outside mathematics; rather, it marks the point at which the realised value becomes part of the encoded history.  A computational account must then specify what information is available before and after this transition.

Selection is not identified with arbitrary randomness.  A random sampler may choose a member of \(\GLUE(x)\), but unpredictability alone does not establish \(\LAc\).  Nor may \(\SELECT\) be replaced by a polynomial-time preference rule fixed by the prior state, since that would violate effective non-anticipability.  In every case, coherence requires
\begin{equation}\label{eq:select-coherence}
  \SELECT(q)\in\GLUE(x).
\end{equation}

The ordinary gluing relation is
\begin{equation}\label{eq:r-glue}
  \RGLUE(q,s)=1 \quad\Longleftrightarrow\quad s\in\GLUE(x).
\end{equation}
Solving the search problem associated with \(\RGLUE\) produces an arbitrary admissible continuation.  It need not produce the realised value \(\SELECT(q)\).  This distinction is essential for the later unique-projection condition.

\begin{rem}[Structural distinction between gluing and selection]\label{lem:glue-select}
If \(\lvert\GLUE(x)\rvert\geq2\), then knowledge of the set \(\GLUE(x)\) alone does not identify \(\SELECT(q)\).  Additional information or an additional operation is required to single out the realised member.
Indeed, membership in a set containing at least two admissible values cannot distinguish the historically realised member from every other member.  This is an elementary structural observation, not a complexity lower bound.
\end{rem}

The minimal architecture is thus
\[
  q=(x,e)\longmapsto F_x\longmapsto\GLUE(x)
  \xrightarrow{\SELECT}s_e.
\]
To connect this architecture with FNP, the realised value must be certified by a fixed relation rather than merely observed informally.

\section{Traceability and retrospective verification}\label{sec:trace}

Admissibility is insufficient to certify actualisation: \(s\in\GLUE(x)\) does not imply \(s=\SELECT(q)\).  We therefore distinguish the witness generated or made available after the act from the uniform predicate that verifies it.

\begin{defi}[Trace and selection certificate]\label{def:trace}
Let \(\operatorname{Actual}(e)=s\) mean that \(s\) is the continuation historically realised at occurrence \(e\).  A trace system is a fixed uniform predicate \(V\) satisfying the historical-adequacy condition
\begin{equation}\label{eq:historical-adequacy}
  \operatorname{Actual}(e)=s
  \quad\Longleftrightarrow\quad
  \exists\tau\;V(q,s,\tau)=1,
  \qquad q=(x,e).
\end{equation}
Thus \(\SELECT(q)\) is the unique continuation historically certified for occurrence \(e\).
The associated selection relation is
\begin{equation}\label{eq:r-select}
  \RSEL(q,s,\tau)=1
  \quad\Longleftrightarrow\quad V(q,s,\tau)=1.
\end{equation}
\end{defi}

The verifier is required to satisfy the following properties on the encoded family.
\begin{enumerate}
  \item \emph{Historical completeness.} For the realised continuation \(s_e\), at least one trace \(\tau_e\) satisfies \(V(q,s_e,\tau_e)=1\).
  \item \emph{Soundness.} If \(V(q,s,\tau)=1\), then \(s\in\GLUE(x)\).
  \item \emph{Fidelity by unique projection.} If both \(V(q,s,\tau)=1\) and \(V(q,s',\tau')=1\), then \(s=s'\).  Several traces may certify one result, but accepted traces cannot certify different results for the same occurrence.
  \item \emph{Retrospective efficiency.} Accepted pairs satisfy \(\lvert s\rvert+\lvert\tau\rvert\leq p(\lvert q\rvert)\) for a fixed polynomial \(p\), and \(V\) runs in polynomial time in the total input length.
\end{enumerate}

The trace must not be a hidden prophecy.  The predicate \(V\) and the identifier \(e\) are fixed independently of the act, but neither \(s_e\) nor a valid trace \(\tau_e\) is included in the prior input \(q\).  The information available at the two stages is therefore
\[
  \text{before: }q=(x,e),
  \qquad
  \text{after: }(x,e,s_e,\tau_e).
\]
Verification uses a relation defined in advance but data made available only after actualisation.  This is the formal asymmetry intended by ``non-anticipable a priori, verifiable a posteriori.''

\begin{prop}[Polynomial verification of singularisation]\label{prop:verification}
Suppose that, for every occurrence input \(q=(x,e)\) in the domain, there is a unique projected result \(s_e\) for which at least one trace \(\tau\) satisfies
\[
  \lvert s_e\rvert+\lvert\tau\rvert\leq p(\lvert q\rvert)
  \quad\text{and}\quad V(q,s_e,\tau)=1,
\]
where \(V\) is polynomial-time decidable, satisfies historical adequacy, and every accepted \(s_e\) belongs to \(\GLUE(x)\).  Then \(\RSEL\) is a polynomially balanced, polynomial-time decidable relation whose accepted witnesses uniquely determine the realised continuation by projection.
\end{prop}

\begin{proof}
Polynomial balance and decidability are assumptions.  Historical completeness gives at least one accepted witness for every occurrence in the domain.  Soundness places its first component in \(\GLUE(x)\), and unique projection forces all accepted witnesses to have the same first component.  Historical adequacy identifies that component with \(\operatorname{Actual}(e)=\SELECT(q)\).
\end{proof}

The proposition does not yet imply a class separation.  It provides the verification side; non-membership in FP will come only from the independent non-anticipability clause in \autoref{def:lac}.

\section{A restricted witness: Tseitin formulas on expanders}\label{sec:tseitin}

The abstract architecture benefits from an explicit model in which locally simple constraints create a genuinely global obstruction.  Tseitin formulas provide such a model.  They are used here to establish a restricted proof-complexity statement, not as a universal lower bound for \(\SELECT\).

Let \(G=(V,E)\) be a connected bounded-degree graph.  Associate a Boolean variable \(x_e\in\{0,1\}\) with each edge \(e\), and a charge \(c_v\in\{0,1\}\) with each vertex \(v\).  The local constraint at \(v\) is
\begin{equation}\label{eq:tseitin-local}
  \bigoplus_{e\ni v}x_e=c_v\pmod 2.
\end{equation}
XOR-ing all vertex constraints cancels each edge variable twice and yields the necessary global condition
\begin{equation}\label{eq:tseitin-global}
  0=\bigoplus_{v\in V}c_v.
\end{equation}
If the total charge is odd, every proper local region may look extendible while the full system is inconsistent.  In the language of \autoref{sec:finite-model}, the vertex constraints define local sections but \(\Gamma(F_G)=\varnothing\).

For a vertex set \(U\subseteq V\), summing the constraints inside \(U\) leaves only boundary variables:
\begin{equation}\label{eq:boundary-parity}
  \bigoplus_{e\in\delta(U)}x_e=\bigoplus_{v\in U}c_v,
\end{equation}
where \(\delta(U)\) is the edge boundary of \(U\).  If \(G\) is an edge expander, then for all \(U\) of size at most \(\lvert V\rvert/2\),
\begin{equation}\label{eq:expansion}
  \lvert\delta(U)\rvert\geq h\lvert U\rvert
\end{equation}
for some constant \(h>0\).  Large partial regions therefore retain large interfaces.  This expansion property motivates the local-to-global interpretation, but no new sheaf-width invariant is asserted here.

Urquhart proved that standard CNF encodings of contradictory Tseitin constraints built from a suitable bounded-degree expander family require exponential-length resolution refutations \cite{Urquhart1987}.  Ben-Sasson and Wigderson later established the general width--size theory for resolution and showed how expansion-based width lower bounds yield size lower bounds \cite{BenSassonWigderson2001}.  For the present paper, only the classical exponential size lower bound is needed:
\[
  \operatorname{Size}_{\Res}(\operatorname{Tseitin}(G_n))
  =2^{\Omega(n)}.
\]

\begin{thm}[Restricted resolution obstruction]\label{thm:tseitin-restricted}
There exists a family of finite local-to-global models obtained from contradictory Tseitin formulas on bounded-degree expander graphs such that the associated CNF formulas require exponential-size resolution refutations.
\end{thm}

\begin{proof}
Take the bounded-degree expander family in Urquhart's construction and an odd charge assignment.  The corresponding Tseitin formulas are unsatisfiable by \eqref{eq:tseitin-global}, and their standard CNF encodings require resolution refutations of size \(2^{\Omega(n)}\) \cite{Urquhart1987}.
\end{proof}

\begin{rem}\label{rem:tseitin-scope}
\autoref{thm:tseitin-restricted} concerns a specified encoding and proof system.  It does not prove that every representation of the same parity structure is hard, nor does an inconsistent Tseitin instance itself model selection, since it has no global continuation to select.  Satisfiable variants or related gadgets may be used to study plural admissible solutions, but the theorem above is retained only as a quantitative laboratory for local-to-global obstruction.
\end{rem}

\section{Effective changes of representation}\label{sec:representation}

A lower bound tied to one syntax may disappear under a better representation.  Gaussian elimination, for example, exposes parity structure that can be difficult for resolution.  A claim of intrinsic non-anticipability must therefore specify which changes of representation are admissible and account for their computational cost.

An encoding scheme \(\mathcal E\) for a process family specifies a polynomial-time decidable language \(I_{\mathcal E}\) of valid occurrence inputs \(q=(x,e)\), the associated admissible sets \(\GLUE_{\mathcal E}(q)=\GLUE_{\mathcal E}(x)\), and a uniform certified-actualisation relation \(V_{\mathcal E}(q,s,\tau)\).  Semantic equivalence alone is too weak for complexity arguments: an equivalent encoding could contain the answer or require superpolynomial work to obtain.  We therefore compare entire encoding schemes uniformly.

\begin{defi}[Admissible effective equivalence]\label{def:effective-equivalence}
Two encoding schemes \(\mathcal E\) and \(\mathcal E'\) are \emph{effectively equivalent}, written \(\mathcal E\sim_{\eff}\mathcal E'\), when there are uniform polynomial-time translations in both directions for instances and continuations, with polynomially bounded size distortion, such that:
\begin{enumerate}
  \item valid instances are mapped to valid instances and the translations recover the same encoded occurrence up to the chosen coding convention;
  \item the continuation translations give mutually inverse correspondences between \(\GLUE_{\mathcal E}(q)\) and \(\GLUE_{\mathcal E'}(T(q))\);
  \item certified actualisation is preserved: for every valid \(q\) and continuation \(s\),
  \[
    \exists\tau\,V_{\mathcal E}(q,s,\tau)=1
    \quad\Longleftrightarrow\quad
    \exists\tau'\,V_{\mathcal E'}(T(q),U_q(s),\tau')=1,
  \]
  where \(T\) and \(U_q\) are the instance and continuation translations.
\end{enumerate}
\end{defi}

This definition does not use \(\SELECT\) as a primitive invariance requirement.  It preserves the admissible sets and the certified-actualisation relation from which the uniquely projected selected continuation is obtained.  It is stronger than bare semantic equivalence and weaker than claiming a categorical equivalence between arbitrary topoi.

\begin{lem}[Invariance of polynomial computability]\label{lem:effective-invariance}
If \(\mathcal E\sim_{\eff}\mathcal E'\) and the selected continuation has a polynomial-time selector under \(\mathcal E'\), then it has a polynomial-time selector under \(\mathcal E\).
\end{lem}

\begin{proof}
Translate an \(\mathcal E\)-instance \(q\) to the \(\mathcal E'\)-instance \(T(q)\), run the assumed selector, and apply the inverse continuation translation.  All three operations are uniform and polynomial-time.  Preservation of certified actualisation and unique projection ensure that the resulting continuation is the one selected for \(q\).
\end{proof}

Consequently, non-computability in polynomial time is preserved under admissible effective equivalence by contraposition.  This gives the robust form of condition (A):
\begin{equation}\label{eq:astar}
\begin{gathered}
  \text{for every }\mathcal E'\sim_{\eff}\mathcal E,\text{ no uniform
  polynomial-time selector under }\mathcal E'\\
  \text{computes the certified continuation}
\end{gathered}
\end{equation}
on the entire encoded family.

No representation-independent width is inferred from this example.  The general argument relies directly on condition (A) and \autoref{lem:effective-invariance}, while \autoref{thm:tseitin-restricted} retains its stated encoding- and proof-system-specific scope.

\section{The search relation and the conditional separation}\label{sec:main-result}

The final step must avoid replacing a physical event by an undefined mathematical oracle.  We therefore formulate the selected continuation through the relation already introduced in \eqref{eq:r-select}.

We use explicit search-problem terminology.  A polynomially balanced relation whose membership predicate is decidable in polynomial time defines a problem in \(\FNPS\).  It belongs to \(\FPS\) when it has a deterministic polynomial-time selector that returns an accepted output on every input in its domain.  TFNP denotes the total relations in \(\FNPS\).  The unqualified class FP is reserved below for ordinary single-valued functions such as \(\SELECT\).  This convention avoids identifying the possibly multivalued certificate \((s,\tau)\) with its uniquely projected value \(s\).

\subsection{The induced search problem}

Let \(I_{\LAc}\subseteq\{0,1\}^*\) be the language of valid encoded occurrence inputs \(q=(x,e)\) of the process family.  We require \(I_{\LAc}\) to be infinite and decidable in polynomial time.  Assume that \(q,s,\tau\) have finite binary encodings and that there is a polynomial \(p\) such that every relevant accepted pair satisfies
\[
  \lvert s\rvert+\lvert\tau\rvert\leq p(\lvert q\rvert).
\]
Assume further that \(\RSEL(q,s,\tau)\) is polynomial-time decidable, historically complete on the domain, sound with respect to \(\GLUE(x)\), and uniquely projected onto \(s\).  Define
\begin{equation}\label{eq:search-relation}
  R(q,y)=1
  \quad\Longleftrightarrow\quad
  y=(s,\tau)\text{ and }\RSEL(q,s,\tau)=1.
\end{equation}
The associated search problem \(\SEARCHSEL\) asks, on input \(q\), for an accepted output \(y=(s,\tau)\).  Under the preceding assumptions, it belongs to \(\FNPS\).

\(\FNPS\) does not require uniqueness of the entire witness.  Several traces may certify one actualisation.  Unique projection onto \(s\), however, is indispensable: it ensures that a solver cannot return an arbitrary member of \(\GLUE(x)\) while satisfying a different certificate.

To avoid an implicit promise problem, extend \(R\) to a relation \(\widehat R\) on all binary inputs by
\begin{equation}\label{eq:total-search-relation}
  \widehat R(q,y)=1
  \quad\Longleftrightarrow\quad
  \begin{cases}
    R(q,y)=1, & q\in I_{\LAc},\\
    y=\bot, & q\notin I_{\LAc},
  \end{cases}
\end{equation}
where \(\bot\) is a fixed one-symbol output.  Because membership in \(I_{\LAc}\) and the relation \(R\) are polynomial-time decidable, \(\widehat R\) is polynomial-time decidable and polynomially balanced.  It is total on \(\{0,1\}^*\): valid instances have an accepted trace, while invalid encodings have the unique default output \(\bot\).  Thus the extended search problem belongs, more precisely, to the total NP search class TFNP and hence to FNP.  On valid instances it coincides with \(\SEARCHSEL\).

\begin{lem}[Functional projection]\label{lem:functional-projection}
Suppose that \(R\) is total on the domain, polynomially balanced, and polynomial-time decidable, and that
\[
  R(q,(s,\tau))=R(q,(s',\tau'))=1
  \quad\Longrightarrow\quad s=s'.
\]
Every algorithm that solves \(\SEARCHSEL\) computes \(\SELECT(q)\) by projecting its output onto the first component.  In particular,
\[
  \SEARCHSEL\in\FPS
  \quad\Longrightarrow\quad
  \SELECT\in\mathrm{FP}.
\]
\end{lem}

\begin{proof}
A solver must return an accepted pair \((s,\tau)\).  By unique projection and historical adequacy, \(s\) is the unique historically certified value for \(q\), namely \(\SELECT(q)\).  Extracting the first component is linear in the output length.
\end{proof}

\subsection{Main theorem}

\begin{thm}[Conditional functional separation]\label{thm:main}
Suppose that there exists an infinite, uniformly encoded family of processes satisfying \(\LAc\) such that:
\begin{enumerate}
  \item its language \(I_{\LAc}\) of valid encodings is polynomial-time decidable;
  \item every accepted continuation--trace pair is polynomially bounded in the occurrence-input length;
  \item the fixed relation \(\RSEL\) is polynomial-time decidable, historically complete on the family, sound, and uniquely projected onto the realised continuation;
  \item effective non-anticipability holds relative to the uniform polynomial-time model, in the representation-robust sense of \eqref{eq:astar}.
\end{enumerate}
Then there exists a total FNP search problem that is not solvable in deterministic polynomial time.  Consequently,
\[
  \LAc\quad\Longrightarrow\quad
  \FPS\neq\FNPS.
\]
\end{thm}

\begin{proof}
By assumptions (1)--(3), the extended relation \(\widehat R\) in \eqref{eq:total-search-relation} is polynomially balanced, polynomial-time decidable, and total on all binary inputs.  Hence its search problem belongs to TFNP and therefore to FNP.

Suppose, for contradiction, that this search problem has a polynomial-time deterministic solver.  On every valid \(q\in I_{\LAc}\), that solver must return an accepted pair \((s,\tau)\), rather than \(\bot\).  By \autoref{lem:functional-projection}, projecting this output computes \(\SELECT(q)\) in polynomial time on the encoded family.  This contradicts assumption (4).

Therefore the extended problem lies in \(\mathrm{TFNP}\setminus\FPS\).  In particular, \(\FPS\neq\FNPS\).
\end{proof}

\begin{cor}[Conditional decision-class separation]\label{cor:p-np}
Under the assumptions of \autoref{thm:main},
\[
  \mathrm{P}\neq\mathrm{NP}.
\]
\end{cor}

\begin{proof}
If \(\mathrm{P}=\mathrm{NP}\), every polynomially balanced NP search relation has a deterministic polynomial-time selector: a witness can be reconstructed by prefix search using polynomially many decision queries \cite{Papadimitriou1994,AroraBarak2009}.  Thus \(\mathrm{P}=\mathrm{NP}\) would imply \(\FNPS=\FPS\).  The contrapositive, together with \autoref{thm:main}, yields the claim.
\end{proof}

\subsection{Exact scope of the result}

The theorem does not establish the existence of human, artificial, or abstract free will.  It does not give an unconditional separation of P and NP.  It proves an implication: if an infinite uniform family of processes satisfies the mathematical definition of \(\LAc\), including the explicit trace and non-anticipability requirements, then polynomially verifiable search differs from polynomially solvable search and hence P differs from NP.

The result is also not intended to hide the separation inside terminology.  The definition supplies a structure beyond the bare existence of a polynomially verifiable search relation with no polynomial-time selector: locally constrained admissibility, a non-singleton globalisation set, historical singularisation, temporal information asymmetry, a uniform retrospective trace, and robustness under effective representation changes.  Nevertheless, condition (A) is a strong computational premise.  The value of \autoref{thm:main} depends on defending that premise as an independently meaningful property of choice processes and, ultimately, on exhibiting a family that satisfies all clauses without presupposing the conclusion.

The roles of the ingredients may be summarised as
\begin{align*}
  \text{topos/sheaf viewpoint}&\longrightarrow \text{finite choice model}
  \longrightarrow \operatorname{CSP}/\GLUE,\\
  \GLUE&\longrightarrow \SELECT
  \longrightarrow \text{trace relation in }\FNPS,\\
  \text{effective non-anticipability}&\longrightarrow \SELECT\notin\mathrm{FP},\\
  \text{functional projection}&\longrightarrow \FPS\neq\FNPS
  \longrightarrow \mathrm{P}\neq\mathrm{NP}.
\end{align*}

\section{Conclusion}\label{sec:conclusion}

We have formalised computational free will as a local-to-global choice structure in which prior determinations constrain several admissible global continuations without supplying an efficient uniform prediction of the one that will be realised.  The distinction between \(\GLUE\) and \(\SELECT\) separates solution admissibility from historical singularisation.  A fixed trace relation makes the realised continuation efficiently certifiable after actualisation, and unique projection makes the resulting search problem functional at the level of the selected value.

Tseitin formulas on expanders show, in a restricted and explicit setting, how simple local constraints can generate a hard global obstruction.  Effective equivalence prevents this illustration from being mistaken for a representation-independent lower bound.  The general theorem instead isolates the precise conditional bridge: existence of a uniformly encoded \(\LAc\) family with polynomial retrospective certification and polynomial-time non-anticipability yields a total FNP search relation with no deterministic polynomial-time selector.

The remaining mathematical frontier is therefore an existence and adequacy problem.  One must determine whether a non-artificial family satisfies all six clauses of \(\LAc\), whether the trace relation can be specified without smuggling the selected value into the prior input, and whether non-anticipability can be established independently of the class separation it entails.  These questions delimit a programme connecting sheaf theory, search complexity, information flow, cognition, and the formal analysis of choice.

\bibliographystyle{alphaurl}
\bibliography{references}

\end{document}